\documentclass{vgtc}                          

\ifpdf
  \pdfoutput=1\relax                   
  \usepackage{graphicx}                
  \DeclareGraphicsExtensions{.pdf,.png,.jpg,.jpeg} 
\else
  \usepackage{graphicx}                
  \DeclareGraphicsExtensions{.eps}     
\fi%

\graphicspath{{figures/}{pictures/}{images/}{./}} 

\usepackage{microtype}                 
\PassOptionsToPackage{warn}{textcomp}  
\usepackage{textcomp}                  
\usepackage{mathptmx}                  
\usepackage{times}                     
\usepackage{cite}                      
\usepackage{tabu}                      
\usepackage{booktabs}                  

\usepackage{dblfloatfix}
\usepackage{mathtools}
\usepackage{amsthm, amssymb}
\usepackage{apxproof}

\newtheorem{lemma}{Lemma}
\newtheorem{observation}{Observation}

\newtheorem{theorem}{Theorem}

\usepackage{thm-restate}

\usepackage{bibunits}
\defaultbibliographystyle{abbrv-doi-hyperref}
\defaultbibliography{references.bib}

\usepackage{wrapfig}

\usepackage{enumitem}

\newcommand{\cparagraph}[1]{\par\vspace{2mm}\noindent\textbf{#1}}

\usepackage{rotating}
\usepackage{float}

\usepackage[ruled,vlined]{algorithm2e}
\SetKwBlock{Repeat}{repeat}{}
\DontPrintSemicolon
\SetArgSty{textnormal} 

\SetCommentSty{mycommfont}
\SetKwInput{Input}{Input}
\SetKwInput{Output}{Output}

\onlineid{0}

\vgtccategory{Research}

\vgtcinsertpkg

\usepackage{subcaption}

\title{ParkView: Visualizing Monotone Interleavings}

\author{Thijs Beurskens\thanks{e-mail: [t.p.j.beurskens\textbar s.w.v.d.broek\textbar a.simons1\textbar w.m.sonke\\\hspace*{14mm}\textbar k.a.b.verbeek\textbar b.speckmann]@tue.nl}\\ %
    \scriptsize TU Eindhoven %
\and Steven van den Broek\footnotemark[1]\\ %
    \scriptsize TU Eindhoven %
\and Arjen Simons\footnotemark[1]\\ %
    \scriptsize TU Eindhoven %
\and Willem Sonke\footnotemark[1]\\ %
    \scriptsize TU Eindhoven %
\and Kevin Verbeek\footnotemark[1]\\ %
    \scriptsize TU Eindhoven %
\and Tim Ophelders\thanks{e-mail: t.a.e.ophelders@uu.nl}\\ %
    \parbox{1.4in}{\scriptsize \centering TU Eindhoven \\ Utrecht University} 
\and Michael Hoffmann\thanks{e-mail: hoffmann@inf.ethz.ch}\\ %
    \scriptsize ETH Zürich %
\and Bettina Speckmann\footnotemark[1]\\ %
    \scriptsize TU Eindhoven %
}

\teaser{
  \centering
  \includegraphics{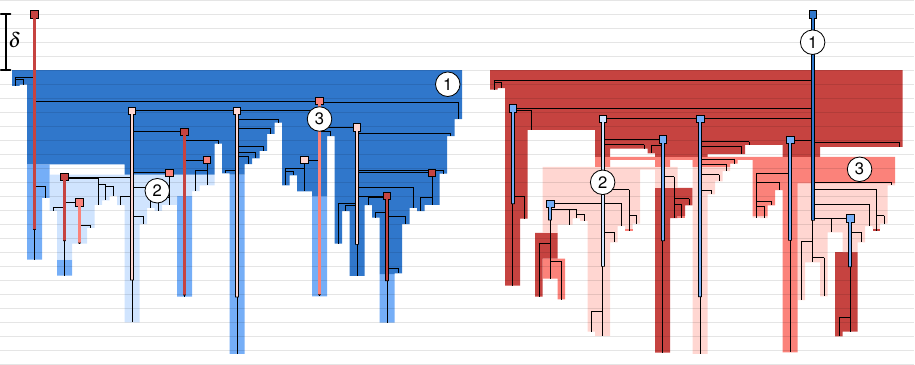}
  \caption{ParkView of a monotone interleaving. The merge trees (47 leaves each) are derived from Timesteps 25 (left) and 26 (right) of the Ionization Front dataset~\cite{pont2022wasserstein}.
  Each merge tree is decomposed optimally into paths; active paths are indicated by vertical colored line segments. 
  Active paths match via color and $x$-order to enclosing shapes (hedges) for branches in the other tree.
  The interleaving maps everything inside a hedge to the corresponding active path (i.e. blue hedge 1 maps to active path 1); hedges and active paths have the same height.
  Hedges can consist of multiple components (hedge 3), which are connected by a line.
  The height (cost) $\delta$ of the interleaving corresponds to the amount that the top active path sticks out of the top hedge.
  }
  \label{fig:teaser}
}

\abstract{
Merge trees are a powerful tool from topological data analysis that is frequently used to analyze scalar fields. 
The similarity between two merge trees can be captured by an interleaving: a pair of maps between the trees that jointly preserve ancestor relations in the trees. 
Interleavings can have a complex structure; visualizing them requires a sense of (drawing) order which is not inherent in this purely topological concept. 
However, in practice it is often desirable to introduce additional geometric constraints, which leads to variants such as labeled or monotone interleavings. Monotone interleavings respect a given order on the leaves of the merge trees and hence have the potential to be visualized in a clear and comprehensive manner.

In this paper, we introduce ParkView: a schematic, scalable encoding for monotone interleavings.
ParkView captures both maps of the interleaving using an optimal decomposition of both trees into paths and corresponding branches. We prove several structural properties of monotone interleavings, which support a sparse visual encoding using \emph{active paths} and \emph{hedges} that can be linked using a maximum of 6 colors for merge trees of arbitrary size. 
We show how to compute an optimal \emph{path-branch decomposition} in linear time and illustrate ParkView on a number of real-world datasets.
} 

\CCScatlist{
  \CCScatTwelve{Human-centered computing}{Visu\-al\-iza\-tion}{Visu\-al\-iza\-tion techniques}{};
  \CCScatTwelve{Networks}{Topology analysis and generation}{}{};
  \CCScatTwelve{Theory of computation}{Design and analysis of algorithms}{}{}
}

 \makeatletter
\def\hyper@natlinkstart#1{%
  \Hy@backout{#1}%
  \hyper@linkstart{cite}{cite.\@bibunitname.#1}%
  \def\hyper@nat@current{#1}%
}

\def\hyper@natlinkbreak#1#2{%
  \hyper@linkend#1\hyper@linkstart{cite}{cite.\@bibunitname.#2}%
}

\def\hyper@natanchorstart#1{%
  \hyper@anchorstart{cite.\@bibunitname.#1}%
}

\def\bibcite#1#2{%
  \@newl@bel{b}{#1}{\hyper@@link[cite]{}{cite.\@bibunitname.#1}{#2}}%
}%

\def\@lbibitem[#1]#2{%
  \@skiphyperreftrue
  \H@item[\hyper@anchorstart{cite.\@bibunitname.#2}%
  \@BIBLABEL{#1}\hyper@anchorend\hfill]%
  \@skiphyperreffalse
  \if@filesw
    \begingroup
      \let\protect\noexpand
      \immediate\write\@auxout{%
        \string\bibcite{#2}{#1}%
      }%
    \endgroup
  \fi
  \ignorespaces
}%

\def\@bibitem#1{%
  \@skiphyperreftrue\H@item\@skiphyperreffalse
  \hyper@anchorstart{cite.\@bibunitname.#1}\relax\hyper@anchorend
  \if@filesw
    \begingroup
      \let\protect\noexpand
      \immediate\write\@auxout{%
        \string\bibcite{#1}{\the\value{\@listctr}}%
      }%
    \endgroup
  \fi
  \ignorespaces
}%

\def\@citex[#1]#2{%
  \let\@citea\@empty
  \@cite{%
    \@for\@citeb:=#2\do{%
      \@citea
      \def\@citea{,\penalty\@m\ }%
      \edef\@citeb{\expandafter\@firstofone\@citeb}%
      \if@filesw
        \immediate\write\@auxout{\string\citation{\@citeb}}%
      \fi
      \@ifundefined{b@\@citeb}{%
        \mbox{\reset@font\bfseries ?}%
        \G@refundefinedtrue
        \@latex@warning{%
          Citation `\@citeb' on page \thepage \space undefined%
        }%
      }{%
        \hyper@natlinkstart{\@citeb}%
            \hbox{\csname b@\@citeb\endcsname}%
        \hyper@natlinkend%
      }%
    }%
  }{#1}%
}%

\makeatother

\begin{document}
\begin{bibunit}
\maketitle

\section{Introduction} 
At the heart of topological data analysis lie so-called topological descriptors: summaries that identify important features of the data.
Topological descriptors are used in various application domains, such as medical imaging \cite{singh2023topological}, manufacturing~\cite{uray2024topological}, and environmental science~\cite{ver2023primer}.
They are generally classified into three types~\cite{yan2021scalar}: set-based descriptors such as persistence diagrams \cite{edelsbrunner2002topological}, complex-based descriptors including Morse-Smale complexes \cite{edelsbrunner2003hierarchical}, and graph-based descriptors such as Reeb graphs \cite{biasotti2008reeb, reeb1946points}.

\addtocounter{figure}{2}

\begin{figure*}[b]
    \begin{subfigure}{0.35\textwidth}
        \centering
        \includegraphics[width=\textwidth]{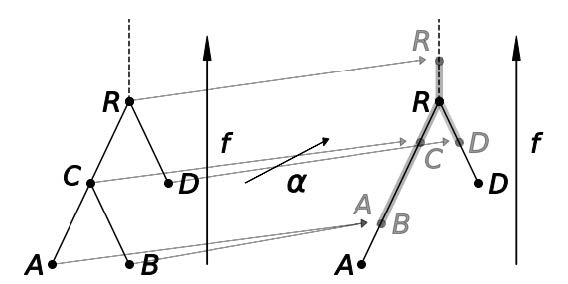}
        \vspace{0mm}
        \caption{Arrows from points to their image~\cite{pegoraro2021graph}}
        \label{fig:vis:arrows}
    \end{subfigure}
    \hfill
    \begin{subfigure}{0.25\textwidth}
        \centering
        \includegraphics[width=\textwidth]{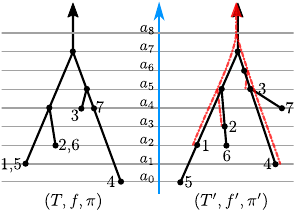}
        \caption{Offset image of tree \cite{gasparovich2019intrinsic}}
        \label{fig:vis:project}
    \end{subfigure}
    \hfill
    \begin{subfigure}{0.3\textwidth}
        \centering
        \includegraphics[]{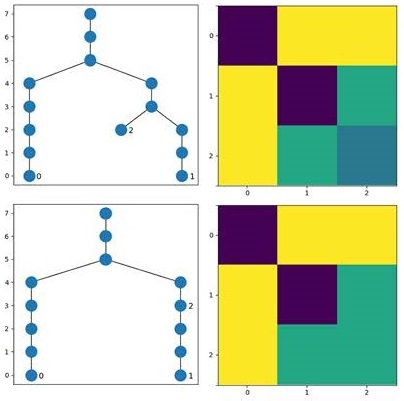}
        \caption{Labeling \cite{curry2022decorated}}
        \label{fig:vis:matrix}
    \end{subfigure}
    \caption{Visualizing interleavings.}
\end{figure*} 

We focus on merge trees, which are a graph-based topological descriptor.
A merge tree captures how the critical features of a scalar field---minima, maxima and saddle points---are connected (\autoref{fig:merge-tree-from-data}).
Merge trees are among the most important tools to support the visualization and analysis of scalar fields. For example, merge trees can be used to track features of a time-varying scalar field~\cite{saikia2017global} and to summarize~\cite{yan2020structural} or detect outliers~\cite{yan2022geometry} in an ensemble of time-varying scalar fields.
See \cite{DBLP:journals/tvcg/LukasczykWWWG24, weiran, DBLP:journals/tvcg/PontT24, qin, wetzels2024exploring, wetzels2024merge} for some very recent additional examples and the survey by Yan et al.~\cite{yan2021scalar} for an extensive overview of merge trees, and other topological descriptors, and their applications for scientific visualization.

\cparagraph{Merge trees and interleavings.} 
There are various similarity measures that are used to compare merge trees, such as the merge tree edit distance~\cite{sridharamurthy2020edit, sridharamurthy2023comparative, wetzels2022branch}, the Wasserstein distances for merge trees~\cite{pont2022wasserstein}, and the merge tree matching distance~\cite{bollen2023computing}.
We focus on the interleaving distance~\cite{gasparovich2019intrinsic, morozov2013interleaving, touli2022fpt} that captures how far two merge trees are from being isomorphic. The interleaving distance has desirable mathematical properties, such as stability and universality~\cite{bollen2021reeb}.
Intuitively, it ``weaves'' the two trees together via two \emph{shift maps} that take points from one tree to points a fixed distance higher in the other tree while preserving ancestry. The two maps together form an interleaving. Two identical trees can be woven together with two horizontal maps; the \emph{height} of an interleaving, which captures the distance between the two merge trees, corresponds to the distance that each shift map has to go up the tree towards the root for two trees that are different from each other.

\addtocounter{figure}{-3}

\begin{figure}[t]
    \centering
    \includegraphics{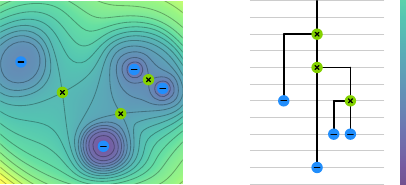}
    \caption{A 2D scalar field with its minima \raisebox{-1pt}{\includegraphics{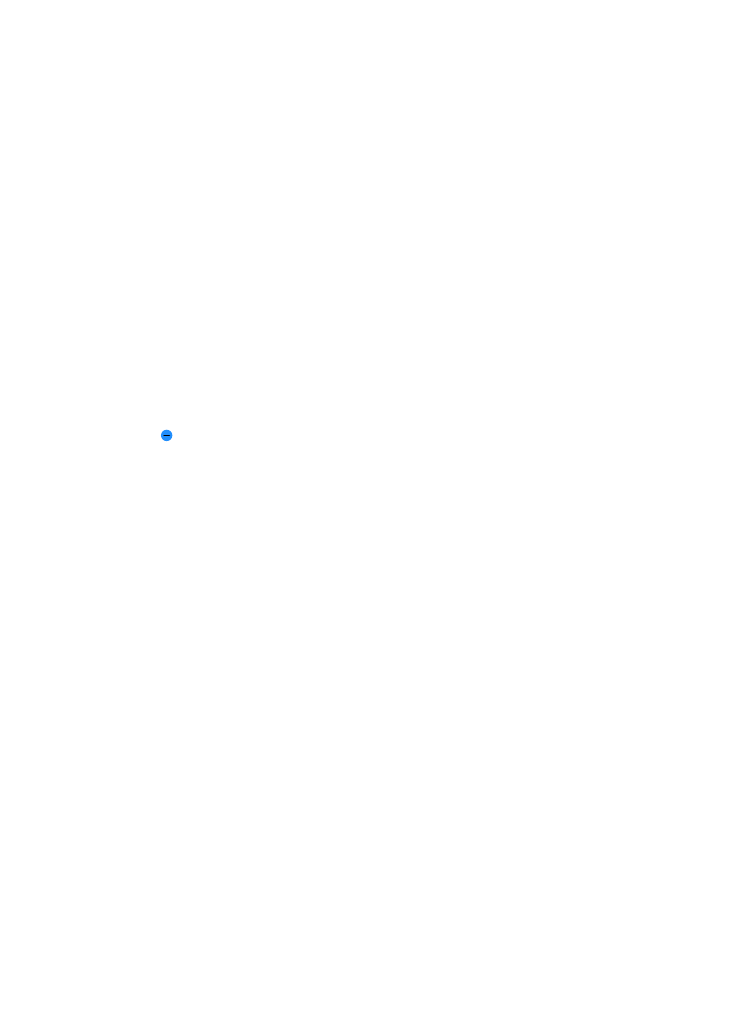}} and saddle points \raisebox{-1pt}{\includegraphics{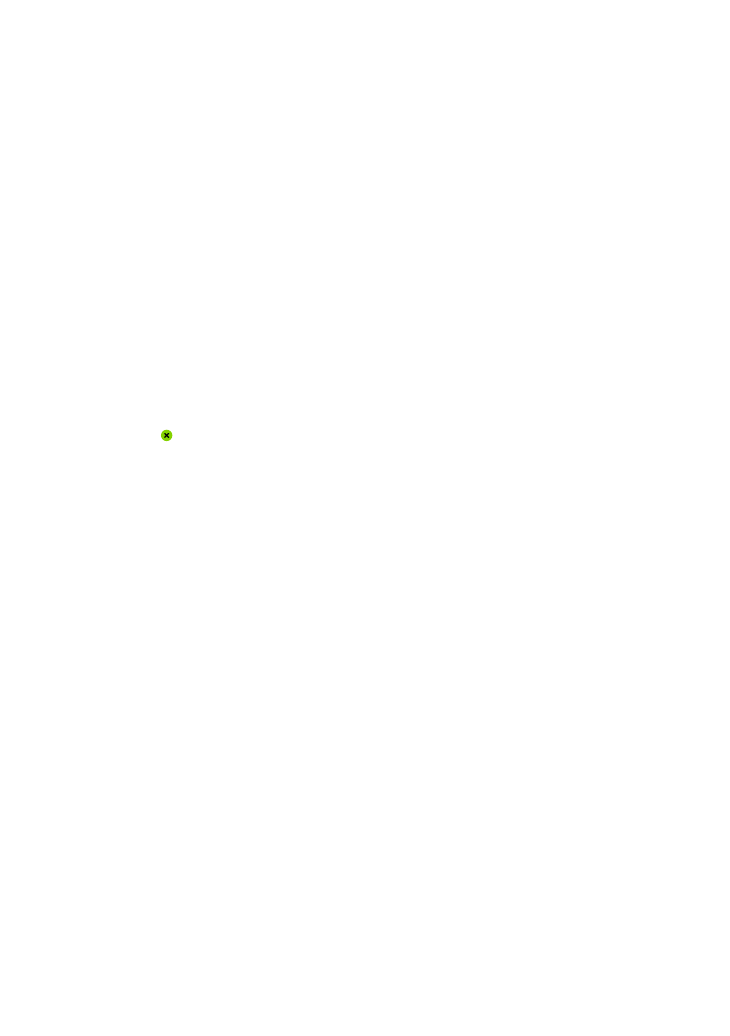}} marked \emph{(left)} and the corresponding merge tree \emph{(right)}.}
    \label{fig:merge-tree-from-data}
\end{figure}

Computing the interleaving distance is NP-hard~\cite{agarwal2018computing}.
Furthermore, the interleaving distance is a purely topological concept. In practice, it is often desirable to introduce additional geometric constraints such as labels or orders, since merge trees frequently arise from spatial terrains.  In addition, to visualize two merge trees together with an interleaving, one needs some sense of (drawing) order to create a meaningful visualization. 
The \emph{labeled interleaving distance}~\cite{gasparovich2019intrinsic} requires a matching of the two merge trees via labels. If such a labeling exists, then the interleaving can be computed efficiently~\cite{munch2019the}.  Yan et al.~\cite{yan2022geometry} recently proposed methods to construct geometry-aware labelings and use them to analyze time-varying data. The \emph{monotone interleaving distance}~\cite{beurskens2025relating} requires only a prior ordering on the leaves of the merge trees that respects the tree structure. Given such an ordering, for example based on the spatial structure of the data, the monotone interleaving distance can be computed efficiently~\cite{beurskens2025relating}.

Formally, a merge tree is a tree $T$ equipped with a function $f$ that assigns a height value to every point of~$T$.
We think of $T$ as a topological space; as such, we refer to not just the vertices, but also each point on the interior of an edge, as a \emph{point} of~$T$.
The highest vertex of $T$ is called the \emph{root}, from which an edge extends upwards to infinity.
The height function $f$ has to be continuous, and strictly increasing along each leaf-to-root path of~$T$.
An \emph{ordered merge tree} is a merge tree equipped with a total order $\sqsubseteq$ on its leaves that respects the tree's structure.
The interleaving distance compares two merge trees $T$ and $T'$ using \emph{$\delta$-interleavings} that consist of two \emph{$\delta$-shift maps}.
A $\delta$-shift map $\alpha$ takes points in $T$ and maps them continuously to points in $T'$ exactly $\delta$ higher.
A $\delta$-interleaving consists of two $\delta$-shift maps---a map $\alpha$ from $T$ to $T'$ and a map $\beta$ from $T'$ to $T$---such that for any point $x \in T$, the point $\beta(\alpha(x))$ is an ancestor of $x$ and for any point $y \in T'$, the point $\alpha(\beta(y))$ is an ancestor of $y$.
\autoref{fig:interleaving-example} shows an example of $\delta$-shift maps and a $\delta$-interleaving.
A $\delta$-shift map or $\delta$-interleaving between two ordered merge trees is \emph{monotone} if it respects the orders of the two trees.
The \emph{(monotone) interleaving distance} is then the smallest $\delta$ for which a (monotone) $\delta$-interleaving exists \cite{gasparovich2019intrinsic, morozov2013interleaving, touli2022fpt}.

\begin{figure}[t]
    \centering
    \hspace*{\fill}
        \centering
        \includegraphics{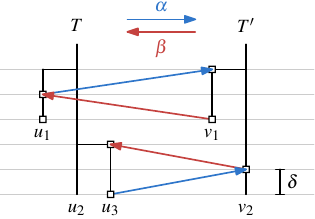}
    \hspace*{\fill}
    \caption{Two $\delta$-shift maps $\alpha$ and $\beta$ that define a monotone interleaving between ordered merge trees $T$ and $T'$. Note that $u_1 \sqsubseteq u_2 \sqsubseteq u_3$ and $v_1 \sqsubseteq' v_2$. The compositions $\alpha \circ \beta$ and $\beta \circ \alpha$ map $u_3$ and $v_2$ to their respective ancestor at height $2\delta$ higher, grid lines at distance~$\delta$.}
    \label{fig:interleaving-example}
\end{figure}

\addtocounter{figure}{1}

\begin{figure*}[t]
    \hspace*{\fill}
    \begin{subfigure}{\columnwidth}
        \centering
        \includegraphics[width=0.9\columnwidth]{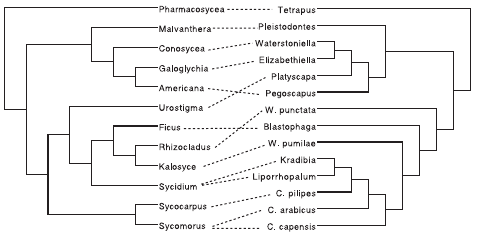}
        \caption{Tanglegram for phylogenetic trees~\cite{DBLP:journals/bioinformatics/ScornavaccaZH11}}
        \label{fig:vis:tanglegram}
    \end{subfigure}
    \hfill
    \begin{subfigure}{\columnwidth}
        \centering
        \includegraphics[width=0.453\textwidth]{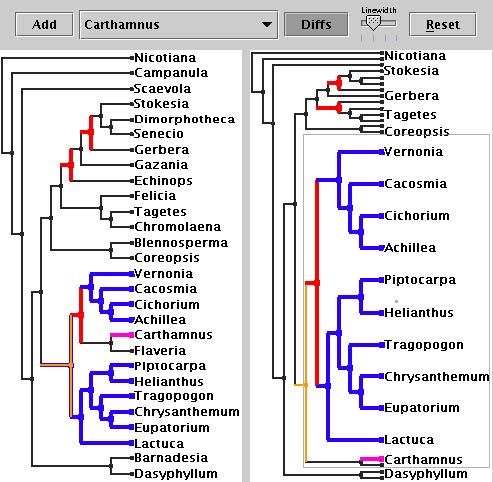}
        \hfill
        \includegraphics[width=0.52\textwidth]{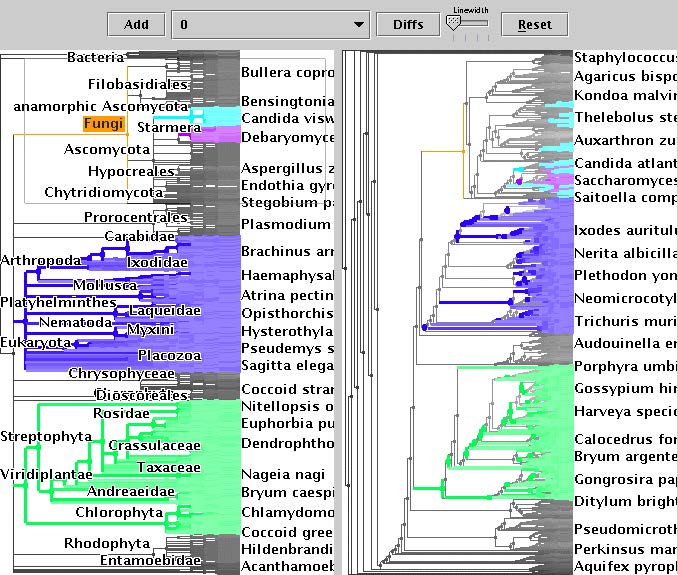}
        \caption{TreeJuxtaposer~\cite{DBLP:journals/tog/MunznerGTZZ03}}
        \label{fig:vis:TreeJuxtaposer}
    \end{subfigure}
    \hspace*{\fill}
    \caption{Visually comparing trees.}
   \label{fig:vis}    
\end{figure*}

\cparagraph{Visualizing interleavings.} Interleavings on merge trees can have a rather complex structure. However, this structure does reveal how much and where the two merge trees differ. As such, visualizations of interleavings could play an important role as part of a visual analytics system, in particular, when combined with brushing and linking to help the user localize the merge trees with respect to the input data. However, currently existing visualizations of interleavings or their constituent shift maps are mostly designed to visually explain the mathematical concept of interleavings on small examples, and not suitable for actual data exploration.

In the following we discuss a set of requirements for mathematically meaningful and effective visualizations of interleavings; we collected these requirements from experts in topological data ana\-lysis. 
The most important requirement is for the visualization to be \emph{complete}. The users should be able to reconstruct both shift maps from the visualization, that is, they should be able to 
\begin{enumerate}[start=0, label={\textbf{R\arabic*}}]
\item \label{req:map}
determine the image of any part of the tree.
\end{enumerate}
A complete visualization by itself is not necessarily effective. The following requirements are intended to highlight the structure of the interleaving and aid the user in recognizing patterns.
Firstly, to evaluate how similar the structures of the two input trees are under the given interleaving, the user should be able to determine the shift $\delta$ of a $\delta$-interleaving easily from the visualization.
In particular, as the value $\delta$ itself does not carry much meaning, it should relate to the input trees themselves.
Secondly, the user should easily be able to identify the parts of a tree that are combined (that is, have the same image) or ignored (that is, are not mapped to) under the interleaving.
To summarize, we require that a user can effectively
\begin{enumerate}[start=1, label={\textbf{R\arabic*}},noitemsep]
\item \label{req:delta}
determine $\delta$ relative to the height of the trees;
\item \label{req:image}
determine which parts of the trees have points mapped to them and which do not;
\item \label{req:same-image}
determine whether points have the same image.
\end{enumerate}
Last but not least, the visualization should scale to ordered merge trees with about 100 leaves.

As mentioned above, currently existing visualizations of interleavings are not designed to support data exploration and mostly do not satisfy these requirements.
Most commonly, a shift map of an interleaving is visualized by explicitly drawing arrows from points in the tree to their image in the other tree (see \autoref{fig:vis:arrows})~\cite{agarwal2018computing, morozov2013interleaving, touli2022fpt}. This visualization might become too complex for large inputs, as a complete visualization may require as many arrows as there are leaves in the trees, and many of the arrows may cross in the visualization.
However, they do make it easy to find the image of a single point.
Another way interleavings have been visualized is by drawing the image of a tree slightly offset on the other tree (see \autoref{fig:vis:project})~\cite{gasparovich2019intrinsic, pegoraro2021graph}.
Such visualizations have low visual complexity and appear to satisfy \ref{req:image} and \ref{req:delta} quite well. However, they have not been designed to satisfy \ref{req:map} and \ref{req:same-image}, even in such a small example.

Every interleaving induces a labeling on the trees. Hence, instead of visualizing an interleaving directly, one could visualize the corresponding labeling instead.  For example, the trees in \autoref{fig:vis:project} are labeled and can be combined with two matrices that describe the heights of the lowest common ancestors of each pair of labels (see \autoref{fig:vis:matrix})~\cite{curry2022decorated}. This visualization is complete, as a label is assigned to every leaf, and the matrix part of this visualization seems to scale well. However, the labelings provide only local information and the behavior of the shift map on parts of the tree without labels has to be derived by the user from those parts that do have labels. Furthermore, integrating the information from the matrices with the trees is non-trivial and both \ref{req:delta} and \ref{req:image} are not met, since the heights of labels do not necessarily carry any relevant information.

An interleaving is a type of matching between two trees. Hence in principle any visualization that matches and compares trees could be used to illustrate interleavings. However, drawings of trees, especially when augmented with matching lines~\cite{DBLP:journals/cgf/HoltenW08, DBLP:journals/bioinformatics/ScornavaccaZH11, wetzels2024merge} or other visual overlays~\cite{DBLP:journals/tog/MunznerGTZZ03}, contain so much information that they quickly become visually too complex for larger inputs (see \autoref{fig:vis}), similarly to the explicit ``arrow drawings'' mentioned above. 

\cparagraph{Contributions and organization.}
In this paper, we introduce ParkView: a schematic and scalable visual encoding for monotone interleavings. 
We prove properties of monotone interleavings and use them to compute a compressed visual encoding of the interleaving, which still retains all relevant information to satisfy Requirements R0--R3.
Specifically, to represent a shift map, ParkView decomposes the two merge trees into few components such that a component in one tree maps entirely to one component in the other tree.
See \autoref{fig:teaser}: the points in the left tree enclosed by shape 1 (a \emph{hedge}) map to the points in the right tree on segment 1 (an \emph{active path}).
The properties of a monotone interleaving allow us to match components left to right, based on the position of the lowest leaf for hedges and the $x$-position for active paths.
Matching components are also assigned the same color.
The drawings of the two shift maps naturally combine and together show the interleaving.

We first define the decompositions mentioned above and detail the visual design of ParkView.
To keep the visual complexity low, we prefer decompositions with few components; in \autoref{sec:guarantees} we define and prove which decompositions are optimal in this sense (\autoref{thm:heavy-path-branch}).
We also use properties of monotone interleavings to prove that three colors suffice to color the hedges such that neighboring hedges have distinct colors (\autoref{thm:3-colorable}).
These proofs lead to algorithms to compute ParkView, which are described in \autoref{sec:algorithms}.
We implemented our algorithms and showcase results for real-world data sets in \autoref{sec:experiments}. 
Our code is openly available.\footnote{\url{https://github.com/tue-alga/visualizing-interleavings}\label{fn:code}} 
We close with a discussion of our results and possible avenues for future work.

\cparagraph{Note on terminology.}
We designed ParkView for monotone interleavings, hence we use ``interleaving'' to mean ``monotone interleaving'' in the remainder of this paper.
The terms ``path'' and ``branch'', along with their ``decompositions'', have been used in the literature with different meanings in different research areas.
In the context of edit distances for merge trees (cf.~\cite{wetzels2022branch}) a branch (decomposition) typically corresponds to our path (decomposition).
Our terminology mostly follows the conventions from the area of graph theory.


\section{Visual Design}
\label{sec:design}

\addtocounter{figure}{1}

\begin{figure*}[b]
    \centering
    \includegraphics{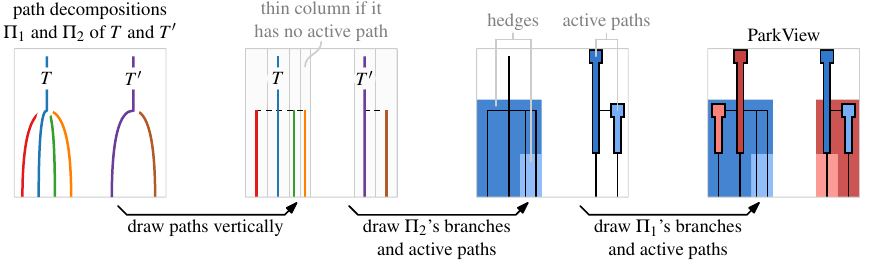}
    \caption{
    We visually encode the branches of a path-branch decomposition via rectilinear enclosing shapes called hedges. A branch in one tree maps to an active path in the other tree; these active paths are drawn as thick colored vertical segments with a square marker at the top.}
    \label{fig:pipeline}
\end{figure*}

Our visual design is built upon an optimal decomposition of the merge trees into paths and branches, which represent those parts of the two trees that are mapped to each other by the interleaving. Below we first describe our decomposition in detail (\autoref{sec:path-branch-decomposition}), and then discuss the visual encoding of all parts (\autoref{sec:drawing-merge-trees}).

\subsection{Path-Branch Decomposition}
\label{sec:path-branch-decomposition}
Our input are two merge trees $T$ and~$T'$ and two shift maps $\alpha$ from~$T$ to~$T'$ and $\beta$ from~$T'$ to~$T$. Recall that each point of the tree has a specific height. In the following we describe the decomposition based on the shift map~$\alpha$; the decomposition based on~$\beta$ is symmetric.

First, we decompose~$T'$ into a \emph{path decomposition}~$\Pi$: a set of height-monotone paths~$\pi$ that each start at a leaf (the \emph{bottom} of~$\pi$) and end at an internal vertex of~$T'$ (the \emph{top} of~$\pi$) or, for one path, at infinity. To make sure the paths of~$\Pi$ are disjoint and exactly cover~$T'$, we consider each path~$\pi$ open at its top; that is, $\pi$ does not contain its top. Alternatively, we can define a path decomposition bottom-up. For a vertex~$v$ of~$T'$, let the \emph{up edge} be the one edge with increasing height incident to~$v$, and let the \emph{down edges} be the other edges incident to~$v$. We now define a path decomposition by selecting, for each internal vertex~$v$, one of the down edges of~$v$ as the \emph{through edge} of~$v$. The path decomposition is then built by starting a path at each leaf of~$T'$, and for each internal vertex~$v$ letting the incoming path from the through edge continue, while the incoming paths from the remaining down edges end at~$v$.

Each path~$\pi \in \Pi$ induces a \emph{branch}~$B_\pi$ in~$T$: the part of~$T$ that $\alpha$ maps to~$\pi$. The branch $B_{\pi}$ can either be empty, or consist of a single connected component (a \emph{simple branch}), or consist of multiple connected components (a \emph{compound branch}) (see \autoref{fig:branches}). 
The complete set of branches~$B_\pi$ forms a decomposition of~$T$, which we call the \emph{branch decomposition} of~$T$. Together, we call the paths in~$T'$ and the branches in~$T$ a \emph{path-branch decomposition} for~$\alpha$.

\addtocounter{figure}{-2}

\begin{figure}[t]
    \hspace*{\fill}
    \begin{subfigure}{0.3\columnwidth}
        \centering
        \includegraphics[page=1]{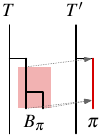}
        \caption{Simple branch}
    \end{subfigure}
    \hfill
    \begin{subfigure}{0.3\columnwidth}
        \centering
        \includegraphics[page=2]{branches}
        \caption{Compound branch}
    \end{subfigure}
    \hfill
    \begin{subfigure}{0.3\columnwidth}
        \centering
        \includegraphics[page=3]{branches}
        \caption{Empty branch}
    \end{subfigure}
    \hspace*{\fill}

    \caption{Examples of branches~$B_\pi$ for a path~$\pi$.}
    \label{fig:branches}
\end{figure}

\addtocounter{figure}{1}
    
A shift map admits many possible path-branch decompositions. Each path-branch decomposition contains the same number of paths, and hence the same number of branches, which is equal to the number of leaves of~$T'$. However, the number of branch components can differ per decomposition. To minimize visual complexity, we aim to construct a path-branch decomposition that minimizes (1) the maximum number of branch components per path and (2) the total number of branch components. In \autoref{sec:guarantees} we prove that there is a form of a heavy-path decomposition that optimizes both criteria simultaneously (\autoref{thm:heavy-path-branch}). Moreover, such an optimal path-branch decomposition can be computed in linear time using a comparatively simple greedy algorithm (see \autoref{sec:algorithms}).

\begin{figure}[t]
    \begin{subfigure}{0.32\columnwidth}
        \centering
        \includegraphics[page=4, width=\textwidth]{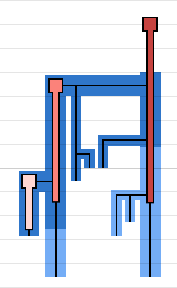}
        \caption{Tight}
        \label{fig:hedge-drawings:tight}
    \end{subfigure}
    \hfill
    \begin{subfigure}{0.32\columnwidth}
        \centering
        \includegraphics[page=2, width=\textwidth]{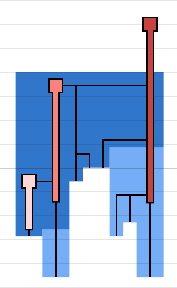}
        \caption{ParkView}
        \label{fig:hedge-drawings:ParkView}
    \end{subfigure}
    \hfill
    \begin{subfigure}{0.32\columnwidth}
        \centering
        \includegraphics[page=3, width=\textwidth]{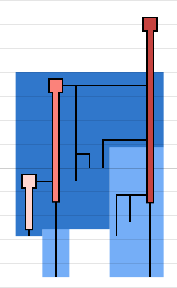}
        \caption{Min-link}
        \label{fig:hedge-drawings:min-link}
    \end{subfigure}
    \caption{Approaches to creating enclosing shapes. The tight approach resembles TreeJuxtaposer~\cite{DBLP:journals/tog/MunznerGTZZ03} (\autoref{fig:vis:TreeJuxtaposer}). We use (b).}
    \label{fig:hedge-drawings}
\end{figure}

\subsection{Visual Encoding}
\label{sec:drawing-merge-trees}
ParkView visualizes an interleaving $(\alpha, \beta)$ between two ordered merge trees by superimposing drawings of path-branch decompositions of its shift maps $\alpha$ and $\beta$. 
Our visual encoding of these path-branch decompositions is most easily understood procedurally; \autoref{fig:pipeline} gives an overview.

\pagebreak
\cparagraph{Path decomposition.} 
The backbone of ParkView is a drawing of the path decompsitions of the trees. 
We draw each path as a thin, black, vertical line segment.
We draw them in a left-to-right order corresponding to the leaf ordering imposed by the tree. 
Therefore, we can think of our drawing as being divided into a series of fairly narrow columns, each containing exactly one path (and hence, exactly one leaf) of the tree.
We connect the paths with horizontal line segments, each of which represents an internal vertex of the merge tree.
As the paths cover the entire merge tree, our representations of the paths and vertices together display the structure of the tree.

Our rectilinear merge tree drawing style is similar to that of Pont et al.~\cite{pont2022wasserstein}. 
In their drawing, low-persistence branches receive little emphasis. 
We similarly de-emphasize parts of the trees: if none of the points in a column are mapped to then we narrow the column.
 
\cparagraph{Active paths.}
Let $B_{\pi}$ be a branch of a path $\pi$.
If $B_{\pi}$ is not empty, then $\alpha$ maps $B_{\pi}$ to a contiguous part of~$\pi$: the \emph{active path} $\pi^*$. 
To visually encode an active path we thicken the corresponding part of the tree and fill it with color, such that it can be visually matched to our encoding of~$B_\pi$ which will have the same color. 
At the top of $\pi^*$, to help the user determine color accurately, we place a square glyph with the same color.
Note that an active path $\pi^*$ always forms the top part of $\pi$; therefore, the square glyph is at the top of $\pi$ as well.

\cparagraph{Hedges.}
We represent each branch~$B_\pi$ by a \emph{hedge}~$H_\pi$: a rectilinear shape enclosing~$B_\pi$. Our design process for hedges was guided by the following three criteria. The hedges should (1) stay close to the tree structure, (2) have low complexity, and (3) have large area. Criterion (1) helps with finding points in the tree (which is a prerequisite for \ref{req:map}, \ref{req:image}, and \ref{req:same-image}). Criteria (2) and~(3) help reduce the visual complexity and allow the user to more easily perceive the hedges' colors. There is a trade-off between the criteria: a tight drawing of the hedges along the tree (\autoref{fig:hedge-drawings:tight}) stays close to the tree structure, but the resulting hedges have small area and high complexity. In contrast, a drawing that minimizes the links used in the union of hedges and the individual hedges (\autoref{fig:hedge-drawings:min-link}) has low complexity and large area, but it may not follow the tree structure very well. We settled on a design in between the two extremes (\autoref{fig:hedge-drawings:ParkView}) that satisfies all criteria reasonably well.

In our design, each hedge is histogram-shaped: it is the union of a set of axis-aligned rectangles called \emph{bars} whose tops are aligned.
We call the height of the highest (lowest) point in a branch $B_\pi$ its \emph{top} (\emph{bottom}) \emph{height}.
The tops of the bars in a hedge $H_\pi$ have height equal to the the top height of $B_\pi$.
A hedge consists of three types of bars, which we call \emph{tree bars}, \emph{fillers}, and \emph{bridges}.
For each path~$\sigma$ that contains points in~$B_\pi$, in the column of~$\sigma$ we add a \emph{tree bar} whose bottom height is the height of the lowest point on $\sigma$ that is in~$B_\pi$. The union of these bars may not be connected as the set of columns of paths~$\sigma$ may not be contiguous. In this case, we connect consecutive leaves in the same branch component by adding \emph{fillers} in the columns between them.
The height of such a sequence of fillers is the smallest height of the two bars they connect (\autoref{fig:drawing-hedges:columns}).
It is not obvious that this connection does not cause overlap between hedges; however, we argue in \autoref{sec:guarantees:coloring} that this is the case.
For a compound branch~$B_{\pi}$, we draw its individual branch components like before, and then between them we add a \emph{bridge}: a horizontal connector at the top of the hedge (\autoref{fig:drawing-hedges:connecting-columns}). The height of the bridge is less than the height of the shortest bar in the hedge.

\begin{figure}[t]
    \hspace*{\fill}
    \begin{subfigure}[t]{0.45\columnwidth}
        \centering
        \includegraphics[page=4]{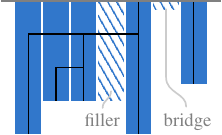}
        \caption{Bars}
        \label{fig:drawing-hedges:columns}
    \end{subfigure}
    \hfill
    \begin{subfigure}[t]{0.45\columnwidth}
        \centering
        \includegraphics[page=5]{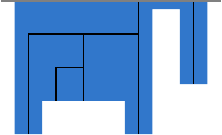}
        \caption{Hedge}
        \label{fig:drawing-hedges:connecting-columns}
    \end{subfigure}
    \hspace*{\fill}
    \caption{A hedge $H_\pi$ consists of three types of bars: (a) tree bars starting at leaves of $B_\pi$; fillers that connect leaves of a component; and bridges that connect different components of $B_\pi$;  (b) the hedge.}
    \label{fig:drawing-hedges}
\end{figure}

\cparagraph{ParkView.} 
The complete ParkView visualization draws the merge trees side-by-side, such that the height difference between a point and its image is~$\delta$. We further add grid lines to help determine~$\delta$ (\ref{req:delta}), interpret heights, and match points to their image (\ref{req:map}). We space these grid lines $\delta$ (or, if $\delta$ is large, a fraction of~$\delta$) apart.

\cparagraph{Properties.}
Our design has several properties that help the user visually match each branch in one tree to its corresponding active path in the other tree.
Firstly, because we respect the leaf ordering of the ordered merge trees, the left-to-right order of the lowest leaf in each hedge matches the left-to-right order of the corresponding active paths.
Secondly, the maximal height of a hedge is equal to the height of the corresponding active path.
Lastly, as mentioned, we use color to match branches and active paths. 
These colors should be such that adjacent hedges have distinct colors, so that the user can distinguish the different hedges. 
Furthermore, in ParkView we aim to use as few different colors as possible to ensure that they can be easily distinguished.
In fact, the hedges in ParkView are 3-colorable; we prove this in \autoref{sec:guarantees:coloring} (\autoref{thm:3-colorable}).
Hence, in ParkView we need only three colors per tree.
We use two distinct hues: red and blue, one for each path-branch decomposition.

\section{Properties of Monotone Interleavings}\label{sec:guarantees}
In this section, we discuss the structural properties of an interleaving that underlie ParkView.
First, we give a specific path-branch decomposition, which we call the \emph{heavy path-branch decomposition}, and we show that it is \emph{optimal}: it minimizes (1) the maximum number of branch components per path and (2) the total number of branch components.
Secondly, we exploit the structure of a shift map and the fact that we use a heavy path-branch decomposition to show a vital property of our hedge design: the set of hedges is 3-colorable, that is, we can always color them using at most three colors such that no two adjacent hedges have the same color.

\subsection{Heavy Path-Branch Decompositions}\label{sec:heavy-path-branch-decomposition}
Let $\alpha$ be a shift map from~$T$ to~$T'$.
As noted before, we can define a path decomposition of~$T'$ by selecting a through edge for each internal vertex~$v$. Let $B_e$ be the part of~$T$ that $\alpha$ maps to the interior of~$e$, and let the \emph{weight} of~$e$ be the number of connected components of $B_e$. We define a \emph{heavy path decomposition} by selecting the through edge of~$v$ to be a down edge of~$v$ with maximum weight.

Next we prove that a heavy path-branch decomposition is optimal.
We refer to the highest edge $\pi$ traverses as its \emph{top edge}. 
We define the \emph{size} of a branch~$B$ as the number of connected components it consists of. We first show that for a given path $\pi$, the size of its induced branch is equal to the weight of $\pi$'s top edge.

\begin{lemma}\label{lem:size-weight}
    Let~$\pi$ be a path with top edge~$e$. Then the size of~$B_\pi$ is equal to the weight of~$e$.
\end{lemma}
\begin{proof}
    Let $v$ be the top of~$\pi$ and let $h \coloneqq f(v) - \delta$. As $e$ is in~$\pi$, we have that $B_e \subseteq B_\pi$. To prove the lemma it hence suffices to argue that each connected component $C$ of~$B_\pi$ contains exactly one connected component of~$B_e$.
    
    To show that $C$ contains at least one connected component of~$B_e$, we show that $C$ contains a point~$x$ in $B_e$. Take any point~$x' \in C$. If $\alpha(x')$ lies in the interior of $e$, then we take $x \coloneqq x'$. Otherwise, we continuously follow the path from $x'$ to the root of~$T$. As $\alpha$ is continuous, the images of the points on the path (in~$T'$) also form a continuous path. Furthermore, as $\alpha$ is a $\delta$-shift map, the images of these points also have a continuously increasing height value. It follows that at some point, we find points whose image is on~$e$. Take such a point~$x$. By definition $x \in B_e$ (and thus also in $B_\pi$). Furthermore, all points between $x'$ and~$x$ on our path map to points on~$\pi$ in~$T'$. Hence, they are all part of $B_\pi$; hence, they are all part of the same connected component of $B_\pi$, namely~$C$.
    
    To show that $C$ contains at most one connected component of~$B_e$, assume for contradiction that there are two distinct connected components $C_1$ and~$C_2$ of~$B_e$ in~$C$. Then, as before, these components contain points $x_1$ and~$x_2$, respectively, at height $h - \varepsilon$ for some $\varepsilon > 0$ chosen such that no vertices of~$T$ have height between $h$ and $h - \varepsilon$.
    Now there is a path~$\rho$ from~$x_1$ to~$x_2$ entirely within~$C$, as $C$ is connected. There also is a distinct path~$\rho'$ from~$x_1$ to~$x_2$ via the lowest common ancestor~$x_3$ in~$T$ of $x_1$ and~$x_2$. Note that~$f(x_3) \geq h$, so $\rho'$ is not entirely within~$C$; that is, $\rho \neq \rho'$. The union of $\rho$ and~$\rho'$ hence contains a cycle, contradicting the fact that $T$ is a tree.
\end{proof}

\noindent
We are now ready to show that any heavy path-branch decomposition is optimal.
\begin{theorem}\label{thm:heavy-path-branch}
    Any heavy path-branch decomposition minimizes the maximum number of branch components per path and the total number of branch components.
\end{theorem}
\begin{proof}
    Let $\Pi$ be a path decomposition.
    Recall that $\Pi$ can be thought of as selecting one through edge for each vertex~$v$ in~$T'$.
    Define the \emph{cost} of~$v$ as the sum of the weights of $v$'s down edges, excluding its through edge. As these edges are exactly the top edges ending at~$v$, by \autoref{lem:size-weight}, the cost of~$v$ is the number of branch components corresponding to the paths ending at~$v$. Then, the sum of costs of all vertices in~$T'$ is the total number of branch components induced by~$\Pi$. This sum is minimized by minimizing the cost for each vertex~$v$, which is achieved by maximizing the weight of its through edge, that is, picking a heavy edge as the through edge. A similar argument holds for minimizing the maximum number of branch components per path.
\end{proof}

\noindent
To obtain a heavy path-branch decomposition, we can thus use a simple greedy algorithm, which we explain in \autoref{sec:algorithms}.

\subsection{Coloring Hedges}\label{sec:guarantees:coloring}
Recall from \autoref{sec:drawing-merge-trees} that each hedge~$H$ is histogram-shaped.
We define the \emph{left} (\emph{right}) \emph{side} of~$H$ as the left (right) side of the leftmost (rightmost) bar of~$H$.
Two distinct hedges are \emph{adjacent} if their boundaries, excluding corners, overlap.
A hedge $P$ is the \emph{parent} of~$H$ if~$P$ is adjacent to the top of~$H$; we call~$H$ a \emph{child} of~$P$.

We show that the set of hedges in ParkView is 3-colorable.
The proof makes use of three properties of our hedge design. Hedges:
\begin{itemize}[noitemsep]
    \item[(i)] are pairwise interior disjoint;
    \item[(ii)] have at most one parent;
    \item[(iii)] have no hedge adjacent to the bottom of their longest bar.
\end{itemize}

\noindent
Our proofs of these properties rely on two observations about the drawing of~$T$ (\autoref{fig:structural-insights}).
Full proofs are in \autoref{app:proof}.

\begin{figure}
    \hspace*{\fill}
    \begin{subfigure}{0.42\columnwidth}
        \centering
        \includegraphics[page=1]{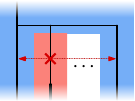}
        \caption{No point can be surrounded by points from another branch}
    \end{subfigure}
    \hfill
    \begin{subfigure}{0.43\columnwidth}
        \centering
        \includegraphics[page=2]{structural-insights}
        \caption{No leaves can be positioned above a horizontal segment}
    \end{subfigure}
    \hspace*{\fill}
    \caption{Structural properties of ParkView.}
    \label{fig:structural-insights}
\end{figure}

\begin{restatable}{observation}{nothorizontallysandwiched}
    \label{lem:not-horizontally-sandwiched}
    A point~$x$ of~$T$ at height~$h$ cannot be surrounded by two points $x_1$ and~$x_2$ at height~$h$ of another branch.
\end{restatable}

\begin{restatable}{observation}{noleafabovehorizontal}
    \label{lem:no-leaf-above-horizontal}
    In the drawing of a tree~$T$, no leaves are positioned vertically above a horizontal segment.
\end{restatable}

\noindent
We now show that ParkView satisfies properties (i)--(iii).

\begin{restatable}{lemma}{propertyproof}
    \label{lem:hedges-disjoint}
	Hedges in ParkView satisfy property (i).
\end{restatable}
\begin{proofsketch}
    For each height in our drawing, we consider a horizontal line at that height. This line intersects a number of points of~$T$, which belong to branches. By \autoref{lem:not-horizontally-sandwiched}, these branches partition the line into disjoint intervals that ``belong'' to each branch.
    We then show that using this procedure, each point in a hedge~$H_\pi$ ``belongs'' to the branch~$B_\pi$. As each point ``belongs'' only to a single branch, it follows that no point can be in more than one hedge.
\end{proofsketch}

\begin{restatable}{lemma}{yetanotherproperty}
    \label{lem:no-multiple-parents}
    Hedges in ParkView satisfy property (ii).
\end{restatable}
\begin{proofsketch}  
    For any hedge~$H_\pi$, we can show that (a) it needs to have a point of~$T$ on the top, which is adjacent to some tree bar in a parent hedge, and (b) any other bars adjacent to the top of~$H_\pi$ need to be part to the same parent hedge.
\end{proofsketch}

\begin{restatable}{lemma}{anotherpropertyproof}
    \label{lem:no-child-on-longest-bar}
	Hedges in ParkView satisfy property (iii).
\end{restatable}
\begin{proofsketch}
    Let $b$ be a longest bar in a hedge~$H_\pi$. We can show that $b$ is a tree bar, because if it were a filler, this would violate \autoref{lem:no-leaf-above-horizontal}. We prove a key property: a tree bar that is a longest bar of its hedge has a leaf of~$T$ on its bottom. Hence, $b$ has such a leaf. Now assume that there is another hedge~$H_\rho$ adjacent to the bottom of~$b$. Then on the top of~$H_\rho$, there has to be a point via which $H_\rho$ connects to the rest of~$T$. As each hedge can have at most one parent (\autoref{lem:no-multiple-parents}) this connection is via a bar~$b'$ of~$H_\pi$. However, then~$b'$ is a longest tree bar. This contradicts our key property that $b'$, being a longest tree bar, has a leaf on its bottom.
\end{proofsketch}

\noindent
We now show that any set of histograms satisfying (i)--(iii) is 3-colorable, from which it follows that our hedges are 3-colorable.

\begin{figure}[b]
    \centering
    \includegraphics{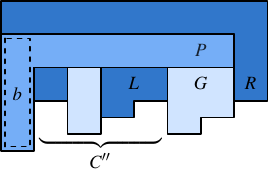}
    \caption{A set of histograms where $P$ is the parent of $G$. The histogram $P$ has a bar~$b$ whose bottom lies below the top of $G$.}
    \label{fig:parent-histogram}
\end{figure}

\begin{theorem}
    \label{thm:3-colorable}
    Any set~$C$ of histograms that satisfies properties (i)--(iii) is 3-colorable.
\end{theorem}
\begin{proof}
    We prove by induction on $n = |C|$.
    For $n = 1$, the theorem trivially holds.
    Now assume that $C$ contains $n+1$ histograms, and let $G$ be a histogram whose top is lowest.
    As the top of any other histogram~$G'$ cannot lie below the top of~$G$, no histogram in~$C$ is adjacent to the bottom side of any bar of~$G$.
    Similarly, there can be only at most one histogram in~$C$ adjacent to the left of~$G$, and there can be only at most one histogram in~$C$ to the right of~$G$.
    Lastly, $G$ can have at most one parent by property~(i), so $G$ has at most three adjacent histograms.

    The set of histograms $C' \coloneqq C \setminus \{G\}$ still satisfies properties (i)--(iii) and has size~$n$.
    Assume (induction hypothesis) that $C'$ is 3-colorable and fix a 3-coloring~$c_1$ for~$C'$. We edit~$c_1$ into a 3-coloring for~$C$.
    If the at most three histograms adjacent to~$G$ use fewer than three colors, we can simply use the third color for~$G$ to obtain a 3-coloring for~$C$.
    Otherwise, denote by $L$ and~$R$ the histograms adjacent to the left and right side of~$G$, and let~$P$ be the parent of~$G$.
    Since $P$, $L$, and~$R$ all have different colors, we can assume without loss of generality that $c_1$ assigns colors 1, 2, and~3 to~$P$, $L$, and~$R$, respectively.
    By property (iii) there is no histogram adjacent to the bottom of a longest bar of~$P$, so we know that $P$ extends below the top of~$G$.
    Thus, $P$ extends below the top of~$G$ either to the left or to the right of~$G$.
    Without loss of generality, assume it extends left of~$G$ and call the rightmost such extending bar~$b$ (\autoref{fig:parent-histogram}).

    Consider the descendants~$C''$ of~$P$ that lie to the left of~$G$ and to the right of~$b$.
    Since $L$ is contained in $C''$, the set $C''$ is nonempty.
    This means that the set $C \setminus C''$ is again a set of histograms that satisfies (i)--(iii) and has size at most $n$, and is hence 3-colorable by the induction hypothesis.
    Let $c_2$ be a 3-coloring of $C \setminus C''$ such that without loss of generality $P$ has color~1 and $G$ has color~3.
    We now define a coloring $c_3$ for~$C$ where the histograms of $C \setminus C''$ take its color from~$c_2$, and the histograms in $C''$ take their color from~$c_1$.

    Note that $G$ and $P$ are the only two histograms of $C \setminus C''$ that are adjacent to histograms in~$C''$.
    So, one of four cases applies to any two adjacent histograms of $C$: (a) both lie in $C \setminus C''$, (b) both lie in $C''$, (c) one is $P$ and the other lies in $C''$ or (d) one is $G$ and the other lies in $C''$ (i.e., the other is $L$).
    For $c_3$ to be a $3$-coloring, it suffices to show that in each case, $c_3$ assigns them distinct colors.
    In case (a), $c_3$ assigns the same distinct colors as $c_1$.
    In case (b), $c_3$ assigns the same distinct colors as $c_2$.
    In case (c), $P$ has color $1$ in both $c_1$ and $c_2$, so $c_3$ again assigns the same distinct colors as $c_2$.
    In case (d), $L$ has color 2 and $G$ has color 3.
\end{proof}

\noindent
Since a hedge is a specific histogram, and our drawing satisfies properties (i)--(iii), we now know that our set of hedges is 3-colorable.

\section{Computing ParkView}\label{sec:algorithms}
ParkView receives as input two ordered merge trees $T$ and $T'$, and a monotone interleaving $(\alpha, \beta)$ between them.
We construct a ParkView visualization using the following four step process:
\begin{enumerate}[noitemsep]
    \item Compute heavy path-branch decompositions for $\alpha$ and $\beta$.
    \item Draw the trees $T$ and $T'$ by drawing their path decomposition.
    \item Draw the branch decompositions as hedges behind the trees, and draw the corresponding active paths on top of the trees.
    \item Color the hedges in both trees, and correspondingly color their active paths in the other tree.
\end{enumerate}
\noindent
The pseudocode we provide in this section is high-level; for details we refer the reader to our implementation\footref{fn:code}.

\cparagraph{Heavy path-branch decompositions.}
We aim to compute a path-branch decomposition with few branch components.
Recall that branches follow uniquely from paths; therefore, our only concern is computing a path decomposition.
As explained in \autoref{sec:path-branch-decomposition}, decomposing a tree into paths is equivalent to choosing at each internal vertex $v$ a down edge that connects to the up edge of $v$ via a path.
In \autoref{sec:heavy-path-branch-decomposition} we define a weight on the down edges of a vertex, and argue that choosing at every vertex the down edge of maximum weight is optimal (\autoref{thm:heavy-path-branch}): the resulting path decomposition minimizes both the total number and the maximum number of branch components.
\autoref{alg:greedy-path-decomposition} computes such a path decomposition in a recursive manner.
In ParkView, when multiple down edges at a vertex have maximum weight, we choose the down edge whose corresponding active path has the lowest start point.

The running time of \autoref{alg:greedy-path-decomposition} is linear in the number of leaves of $T$ and $T'$.
We can for example represent the shift map as an array such that the image under $\alpha$ of a vertex in $T$ can be determined in constant time.
Then by iteratively traversing $T$ from each leaf to its root (and stopping a traversal when encountering an edge that has been traversed before), one can in linear time determine the weights of the edges in $T'$.
After precomputing these edge weights, the recursion then takes linear time.

\begin{algorithm}[tbh]
\Input{Merge trees $T$ and $T'$ and a shift map $\alpha$ from $T$ to $T'$.}
\Output{A heavy path-branch decomposition of $\alpha$.}
\BlankLine

for an edge $e$ of $T'$, let $B_e$ be the part of~$T$ that $\alpha$ maps to the interior of~$e$\;
by simultaneously traversing the trees---following shift map~$\alpha$---pre-compute values $|B_e|$ for each edge~$e$ of $T'$\;

\BlankLine
\SetKwFunction{decompose}{decompose}
\SetKwProg{func}{Function}{}{}
\func{\decompose{$v$}}{
    \If{$v$ is a leaf} {
        \Return $\{ \text{ path from } v \text{ traversing its up edge } \}$\;
    } \Else {
        initialize empty path decomposition $\Pi \gets \emptyset$\;
        \For {child $c$ of $v$} {
            add $\decompose{c}$ to $\Pi$\;
        }
        let $e$ be a down edge at $v$ with $|B_e|$ maximum\;
        let $\pi \in \Pi$ be the path that traverses $e$\;
        extend $\pi$ along the up edge of $v$\;
        \Return $\Pi$\;
    }
}{}
\BlankLine

\Return \decompose{\text{\upshape root of }$T'$}\;

\caption{Heavy Path-Branch Decomposition}
\label{alg:greedy-path-decomposition}
\end{algorithm}

\cparagraph{Drawing and coloring.}
After the path-branch decompositions of the two shift maps have been computed, we first draw the ordered merge trees as described in \autoref{sec:drawing-merge-trees}.
Recall that we draw ordered merge trees rectilinearly by drawing each path in the corresponding path decomposition vertically in its own column.
We make columns with an active path in them wider to emphasize important parts of the tree and improve the scalability of the visualization.
We draw the branches and their corresponding active paths incrementally in a top-down fashion.
\autoref{alg:hedges} describes this drawing algorithm, which simultaneously colors the hedges and corresponding active paths. 
As a last step, we create grid lines and space them by some factor of the interleaving height $\delta$.
We draw the grid as black lines with high transparency on top of the hedges, but behind the tree.

\begin{figure*}[t]
    \centering
    \includegraphics[]{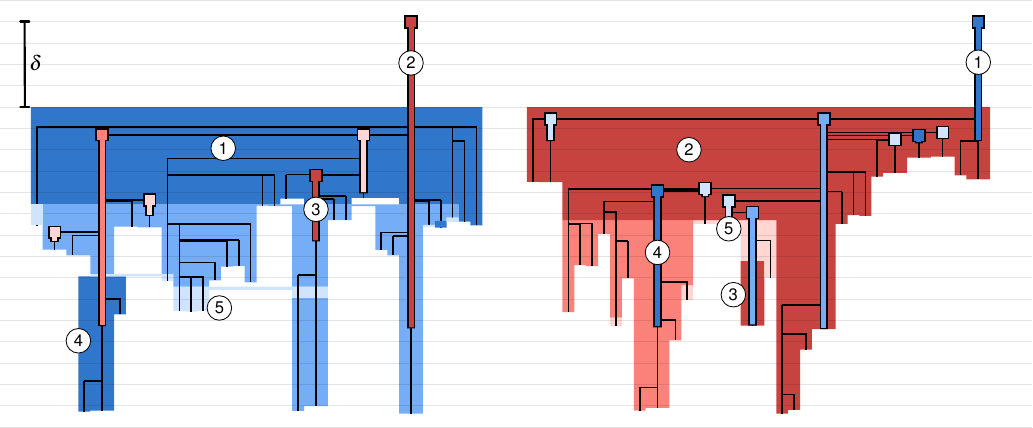}
    \caption{ParkView of a monotone interleaving. The merge trees (32 leaves left, 29 leaves right) are derived from Timesteps 127 (left) and 176 (right) of the Ionization Front dataset \cite{pont2022wasserstein}, using a persistence simplification threshold of 0.05. Grid lines at distance of $\frac{1}{4}\delta$.}
    \label{fig:ionization_main_example}
\end{figure*}

\begin{algorithm}[tbh]

\SetKwFunction{DrawHedge}{DrawHedge}

\Input{A heavy path-branch decomposition $\Pi$ of a monotone shift map $\alpha$ from $T$ to $T'$.}
\BlankLine

sort the paths in $\Pi$ on descending height of their top\\
\BlankLine
\For{path $\pi$ in $\Pi$} {
 draw the active path $\pi^*$ on top of $T'$ \textit{(see \autoref{sec:drawing-merge-trees})}\;
 draw the hedge $G$ representing the branch corresponding to $\pi$ behind the drawing of tree $T$ \textit{(see \autoref{sec:drawing-merge-trees})}\;
 find left $L$ and right $R$ hedges adjacent to $G$\;
 let $P$ be the parent hedge of $G$\;
 \If{$L, R$ and $P$ exist and all have distinct colors}{
    let $b$ be the bar of $P$ closest to $G$ that extends below the top of $G$\;
    let $H_1, \dots, H_k$ be the hedges between $G$ and $b$\;
    let $l$ and~$r$ be the colors of $L$ and~$R$\;
    for the hedges $H_i$, swap the colors $l$ and~$r$\;
 }
 color $G$ and $\pi^*$ with the first color not used by $L, R$ and $P$\;
}
\caption{Draw a Heavy Path-Branch Decomposition}
\label{alg:hedges}
\end{algorithm}

\section{Visual Exploration}\label{sec:experiments}

In this section we illustrate ParkView 
on several real-world merge trees and monotone interleavings.
We first describe our data and preprocessing steps.

\cparagraph{Datasets.}
We use two datasets: 2008\_ionization\_front\_2D (\emph{Ionization Front}) and 2014\_volcanic\_eruptions\_2D (\emph{Volcanic}) \cite{pont2022wasserstein}.
The Ionization Front dataset contains scalar fields that represent the density of shadow instability derived from a simulation of ionization front propagation and comprises 16 time steps. 
%
%
The Volcanic dataset contains scalar fields that represent the sulfur dioxide concentration after a volcanic eruption, obtained by satellite imaging and comprises 12 time steps.

\cparagraph{Pipeline.}
We construct ordered merge trees from the data as follows.
First, we use the Topology Toolkit (TTK)~\cite{ttk} to simplify the scalar fields such that the corresponding merge trees do not have shallow branches; that is, we remove extrema pairs with low persistence~\cite{DBLP:journals/tvcg/LukasczykGMT21}.
We then extract the merge tree from the simplified scalar field.

Since ParkView visualizes monotone interleavings of ordered merge trees, we have to imbue the merge trees with an order. Ideally, the orders for different merge trees in a sequence are stable (they do not change much if the trees/data do not change much) and meaningful for the application at hand. Finding such orders is an interesting open question. For our proof-of-concept, we construct orders via a space-filling curve. Specifically, we first assign a total order to the leaves using the curve and then use a bottom-up approach to make the order consistent with the tree structure by ordering children at internal vertices by their first descendant leaf.

To compute a monotone interleaving we use the relation between the interleaving distance and the Fréchet distance~\cite{beurskens2025relating}.
Specifically, an ordered merge tree induces a 1D curve via an in-order tree traversal, starting and ending at the root of the tree.
We compute a Fréchet matching between the two induced 1D curves, and obtain an interleaving from this matching.
To compute ParkView, we use the steps as outlined in \autoref{sec:algorithms}.
We implemented this pipeline using Python and Kotlin; our code is openly available.\footref{fn:code}

\cparagraph{Showcase.} 
In the following consider \autoref{fig:ionization_main_example}. The height $\delta$ of the visualized interleaving can be easily identified, regardless of the complexity of the interleaving (\ref{req:delta}), as it corresponds to the distance from the highest square glyph to the first hedge below.
The active paths, shown as thick vertical segments, show which parts of the merge trees are mapped to from the other tree (\ref{req:image}).
To determine which hedges map to which active path, we can leverage several properties of ParkView. Consider path~3 in \autoref{fig:ionization_main_example}. 
From its color, we know it maps to hedge 2 or 3 in the right tree. 
As path~3 is left of path~2, it must map to hedge~3 as it the leaf it encloses is left of the lowest leaf hedge~2 encloses.

In addition to leaf order, there are other properties that help identify the mapping, such as the heights of active paths and hedges. The height of the tallest bar in a hedge corresponds to the height of the active path to which it maps. By comparing heights, it becomes clear that hedge 3 maps to path 3.

Additionally, vertical offset provides another means of identification: the top of each hedge is exactly $\delta$ lower than the top of the active path it maps to. For example, when examining path 4 in the right tree, we can see that it is mapped to by hedge 4 in the left tree based on the vertical offset.
By combining these four properties---leaf order, color matching, height correspondence, and consistent vertical offset---we can uniquely determine which active paths in one tree are mapped to by hedges in the other tree (\ref{req:map}). 

\begin{figure*}[t]
    \centering
    \includegraphics[width=\textwidth]{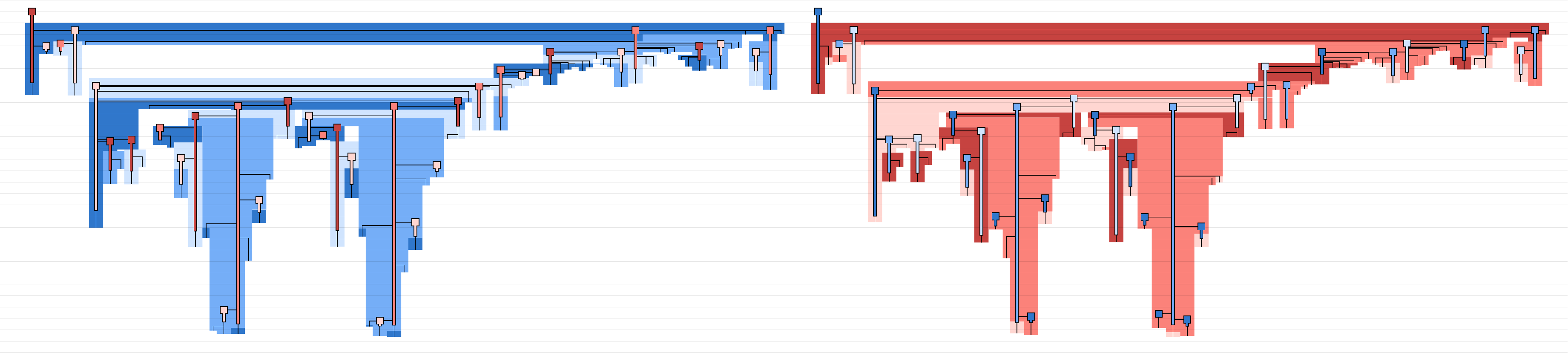}
    \caption{ParkView of a monotone interleaving. The merge trees (73 leaves left, 71 leaves right) are derived from Timesteps 177 (left) and 178 (right) of the Ionization Front dataset \cite{pont2022wasserstein}, using a persistence simplification threshold of 0.01. Grid lines at distance of $\delta$.}
    \label{fig:ionization_large}
\end{figure*}

\begin{figure*}[b]    \centering
    \includegraphics[width=\textwidth]{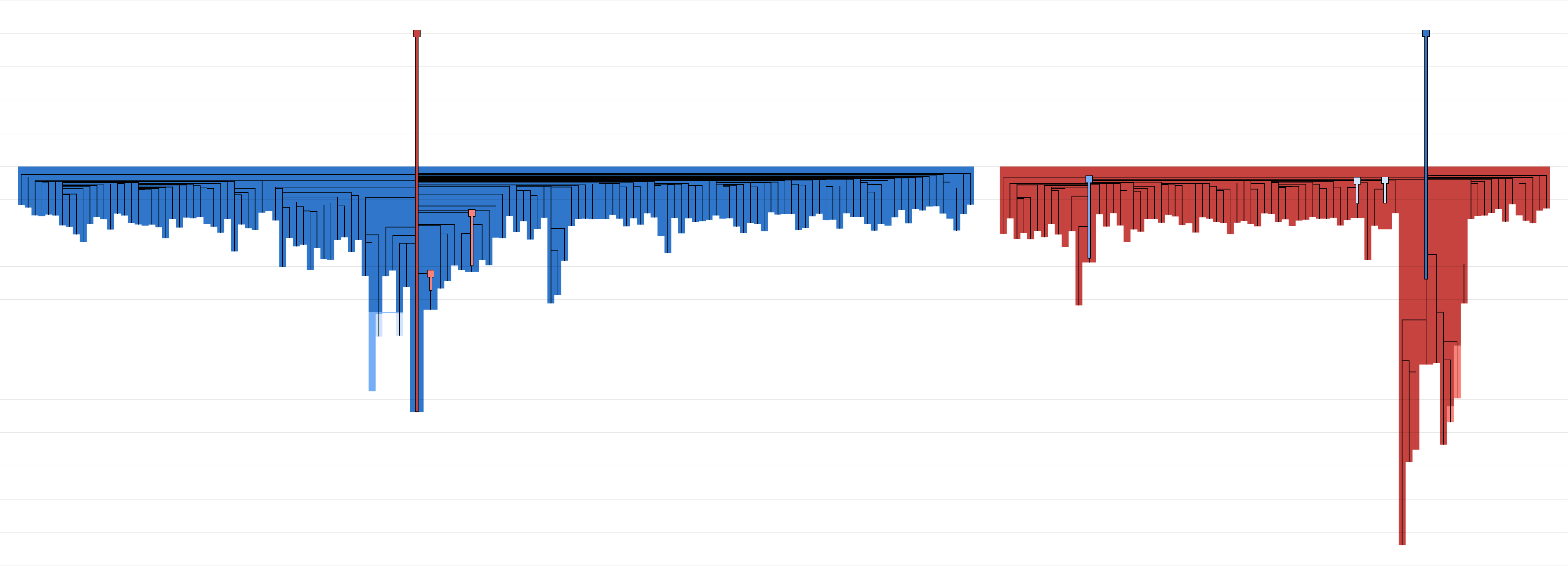}
    \caption{ParkView of a monotone interleaving. The merge trees (136 leaves left, 73 leaves right) are derived from Timesteps 151pm (left) and 156pm (right) of the Volcanic dataset \cite{pont2022wasserstein}, using a persistence simplification threshold of 1.0. Grid lines at distance of $\frac{1}{4}\delta$.}
    \label{fig:volcanic_large}
\end{figure*}

\autoref{fig:ionization_main_example} also shows how columns without an active path take up less horizontal space, which emphasizes parts of the tree that are mapped to (\ref{req:image}). 
This column compression not only highlights relevant aspects of the interleaving, but also decreases the width of both trees. 
This space-saving feature is critical to enable users to view both trees clearly when mapping hedges from one tree to active paths in the other tree.

\cparagraph{Interpreting ParkView.}
A time series of real-world scalar fields typically does not change too drastically over time; the number of extrema and the scale---the minimum and maximum height of points on the scalar field---stay roughly the same.
Therefore, we can expect that the merge trees constructed from such data have roughly the same number of leaves at similar heights.
Under these conditions, we can use ParkView to identify patterns in the given interleaving.
In particular, besides the relative value of~$\delta$, the number of active paths, the size of the hedges, and the drawings of the trees offer an indication of how well two trees compare.
In \autoref{fig:ionization_large} for instance, ParkView contains many active paths, the hedges are relatively small and the trees have a similar structure.
This indicates that the trees, and in addition the corresponding data points, are quite similar.
In \autoref{fig:volcanic_large}, on the other hand, ParkView contains only few active paths, the hedges are large, and the tree drawings are very different.
This indicates that the merge trees are less alike.

\cparagraph{Scalability.} 
We briefly discuss the scalability of ParkView.
When the number of leaves increases and there are relatively few active paths, column compression significantly reduces the horizontal space ParkView uses; see for instance \autoref{fig:volcanic_large}. 
As there are few active paths, matching a hedge to its corresponding active path using the left-to-right order is still feasible. 
If there are many active paths, as in \autoref{fig:ionization_large}, this becomes more difficult. 
However, in such cases, the trees are more alike and often the hedge that corresponds to a path $\pi$ has a similar shape as the hedge that encloses $\pi$, providing an additional cue to match hedges and paths.

As we push scalability limits of ParkView it becomes increasingly clear that some form of tree simplification is necessary to show large trees clearly. The supplementary material includes an example with trees containing over 900 leaves; here inactive paths consume excessive space and obscure the underlying hedges, making the visualization difficult to read. ParkView for trees of this size can still be computed in a matter of seconds.

\cparagraph{More than 3 colors.} 
We prove in \autoref{thm:3-colorable} that 3 colors are sufficient to color the hedges in ParkView. By design, hedges and their corresponding active paths follow the same left-to-right order and hence the matching can be deduced uniquely from the order. The coloring is an additional aid for the user and further serves to visually separate hedges. It might be beneficial for certain applications to use more than the minimum number of 3 colors. Our algorithm supports any number of colors. In the supplementary material we present two variants of \autoref{fig:ionization_large}. One uses the same hue values with six lightness values per tree, and one uses six different hues per tree. We note that with increasing number of colors it necessarily becomes more difficult to distinguish colors.

\section{Conclusion}\label{sec:conclusion}
We introduced ParkView: a compact and scalable visual encoding for monotone interleavings that uses two decompositions of the trees; one into paths, and the other into branches.
ParkView uses a path-branch decomposition that minimizes the number of branch components.
We exploit the properties of this path-branch decomposition, the monotone interleaving, and our approach to drawing branches to show that three colors suffice to color ParkView such that touching hedges have distinct colors.
ParkView shows both shift maps of the interleaving simultaneously by superimposing drawings of branches and the active paths they map to.

As the size of the merge trees increases, non-active paths take up too much visual space. We plan to investigate options to further compress the tree, while still satisfying some form of requirement \ref{req:map}. One possible avenue to explore is interactivity, with parts of the trees folding in and out on demand.

ParkView is designed to compare exactly two merge trees. To use it for the analysis of whole time series or ensembles, further extensions will be necessary. We can imagine solutions that use one reference tree or overlay multiple hedges, however, it is not obvious how to do so without increasing the visual complexity too much.

When using ParkView as part of a visual analytics system on real-world data, we need to construct an order for each merge tree. As discussed in \autoref{sec:experiments}, creating stable orders, that are meaningful for the application at hand, is non-trivial. For spatial-temporal data we can imagine that orderings based on spatial features might be suitable; however, this question merits further investigation. 

ParkView visualizes monotone interleavings of ordered merge trees. 
It is unclear if a similar visualization for more general interleavings exists. We can compute an optimal path-branch decomposition for arbitrary interleavings, but it is unclear how to base a visualization on this decomposition. 
Any drawing of a tree induces an order of its leaves, which implies structure that might not be present. 
Furthermore, without the ability to match branches and paths from left to right, the interleaving will need to be indicated more explicitly, creating visual clutter. Hence visualizing general merge trees and interleavings remains a challenging open problem. 

The (monotone) interleaving distance is a bottleneck measure which forces matchings to always increase by $\delta$. If the distance between two trees is high, then the interleaving distance does not capture the similarity between subtrees which are closer than $\delta$. ParkView still allows some visual matching in such cases, but ultimately, an adaptive, more local version of the interleaving distance would give better insights in the similarity of merge trees. Developing such an improved measure is another challenging open problem.

\acknowledgments{
Research on the topic of this paper was initiated at the 7th Workshop on Applied Geometric Algorithms (AGA 2023) in Otterlo, The Netherlands.
Thijs Beurskens, Willem Sonke, Arjen Simons, and Tim Ophelders are supported by the Dutch Research Council (NWO) under project numbers OCENW.M20.089 (TB, WS), VI.Vidi.223.137 (AS), and VI.Veni.212.260 (TO).
}

\bibliographystyle{abbrv-doi-hyperref}

\putbib
\end{bibunit}
\makeatletter
\begin{bibunit}

\begin{figure*}[t!]
\begin{minipage}{\textwidth}
\makeatletter
\begin{center}%
           {\sffamily\ifvgtcjournal\huge\else\LARGE\bfseries\fi%
	      \vgtc@sectionfont
	     ParkView: Visualizing Monotone Interleavings \\ {\mdseries \Large Supplementary Material} \par}%
       \end{center} 
\makeatother
\end{minipage}
\end{figure*}
\setcounter{figure}{0}
\setcounter{table}{0}

\newpage

\appendix 

\section{Omitted Proofs from \autoref{sec:guarantees}}
\label{app:proof}

Recall that an ordered merge tree~$T$ is equipped with an order on its leaves, denoted~$\sqsubseteq$, that respects the tree structure. To prove the observations and lemmas from~\autoref{sec:guarantees}, it is helpful to provide a more formal definition. In fact, there are two equivalent definitions~\cite{beurskens2025relating} which we both use:
\begin{itemize}
    \item Firstly, an ordered merge tree is a merge tree equipped with a total order~$\sqsubseteq$ on its leaves such that for all leaves $x_1, x_2, x_3$ with $x_1 \sqsubseteq x_2 \sqsubseteq x_3$, it holds that $x_2$ is in the subtree of~$T$ rooted at the lowest common ancestor of $x_1$ and~$x_3$.
    \item Alternatively, an ordered merge tree is a merge tree equipped with a set of total orders~$\leq_h$ (one for each height~$h$) such that for each points $x_1$ and~$x_2$ at height~$h$ with $x_1 \leq_h x_2$, for all $h' > h$ we have $x_1|^{h'} \leq_{h'} x_2|^{h'}$ where $x_1|^{h'}$ and~$x_2|^{h'}$ are the ancestors of $x_1$ and $x_2$ at height~$h'$.
\end{itemize}
Theorem~3.1 of~\cite{beurskens2025relating} proves that these definitions are equivalent.

Additionally, we need a more formal definition of the monotonicity of a shift map. Let $T$ and~$T'$ be ordered merge trees with sets of total orders $\leq_h$ and~$\leq_h'$, respectively. A shift map~$\alpha \colon T \to T'$ is monotone if, for all heights~$h$ and all points $x_1$ and~$x_2$ at height~$h$, we have if $x_1 \leq_h x_2$, then $\alpha(x_1) \leq'_{h + \delta} \alpha(x_2)$.

\nothorizontallysandwiched*
\begin{proof}
    For contradiction, assume that there are three points $x_1$, $x$, and $x_2$ (in this left-to-right order) at height~$h$, where $x_1$ and~$x_2$ lie in branch~$B_\pi$ while $x$ lies in some other branch~$B_\rho$. Consider the total order~$\leq_h$ (in the second definition). As our drawing respects this order, we know $x_1 \leq_h x \leq_h x_2$. Moreover, as $x_1$ and~$x_2$ are in $B_\pi$, their images $\alpha(x_1) = \alpha(x_2)$ lie on~$\pi$; similarly, $\alpha(x)$ lies on~$\rho$. By monotonicity of~$\alpha$, however, $\alpha(x_1) \leq_{h +\delta}' \alpha(x) \leq_{h +\delta}' \alpha(x_2)$. This implies that $\alpha(x_1) = \alpha(x) = \alpha(x_2)$ which is a contradiction with the fact that they lie on distinct paths $\pi$ and~$\rho$.
\end{proof}

\noleafabovehorizontal*
\begin{proof}
    For contradiction, assume that some leaf~$x$ is positioned vertically above a horizontal segment~$s$, where $s$ represents the internal vertex~$v$ of~$T$. Let $l$ and~$r$ be the leftmost and rightmost leaves of the subtree rooted at~$v$; clearly, $l$ is to the left of~$x$ and~$r$ is to the right of~$x$. Because the column order of our drawing corresponds to the order~$\sqsubseteq$ of~$T$ (in the first definition), we have $l \sqsubseteq x \sqsubseteq r$. Then by the definition, $x$ has to be in the subtree rooted at the lowest common ancestor of $l$ and~$r$, which is~$v$. This is however not the case, as $x$ is higher than~$v$. Contradiction.
\end{proof}

\noindent
For convenience, we restate the three properties from~\autoref{sec:guarantees:coloring}. Hedges:
\begin{itemize}[noitemsep]
    \item[(i)] are pairwise interior disjoint;
    \item[(ii)] have at most one parent;
    \item[(iii)] have no hedge adjacent to the bottom of their longest bar.
\end{itemize}
For the purposes of the following proofs, we define a \emph{horizontal (vertical) segment} of a branch $B$ (or hedge~$H$) as a horizontal (vertical) line segment of the drawing of~$T$ that is completely contained in~$B$ (or~$H$).

\propertyproof*
\begin{proof}
    Let $h$ be some fixed height. Consider all points of~$T$ at height~$h$; we call these points \emph{cutpoints}. We consider the set of branches~$\mathcal{B}$ that contain one or more cutpoints. For each branch $B \in \mathcal{B}$, let the \emph{foliage} of~$B$ be the set of leaves in subtrees rooted at cutpoints in~$B$. For each branch~$B \in \mathcal{B}$, define its \emph{range} as the set of consecutive columns in the drawing between the leftmost and rightmost leaves in the foliage of~$B$. Every column~$c$ is contained in the range of at most one branch. To see this, consider the order~$\leq_h$, which determines the left-to-right order on the cutpoints. Correspondingly, we get a left-to-right order~$\leq_h$ on the branches in~$\mathcal{B}$ by \autoref{lem:not-horizontally-sandwiched}: two cutpoints of any given branch cannot have a cutpoint of another branch between them. This branch order also determines the order of the leaves: for two branches $B_1 \leq_h B_2$, any leaf in the foliage of~$B_1$ comes before any leaf in the foliage of~$B_2$ in the $\sqsubseteq$ order.
    
    Now, let~$p$ be a point in some hedge~$H_\pi$. Assume that $p$ is in column~$c$ at height~$h$. We show that $B_\pi$ has $c$ in its range. If $p$ is in a tree bar of~$H_\pi$, then $c$ contains a leaf in the foliage of~$B_\pi$; it follows that $c$ is in the range of~$B_\pi$. Otherwise, $p$ is in a filler or a bridge of~$H_\pi$. In both cases, there are tree bars to the left and to the right of~$c$, which both contain a leaf in the foliage of~$B_\pi$; it again follows that $c$ is in the range of~$B_\pi$.
    
    Assume for contradiction that some point~$p$ in the drawing (say, in column~$c$) is in two distinct hedges $H_\pi$ and~$H_\rho$. Then $B_\pi$ and~$B_\rho$ both have $c$ in their ranges, contradicting the fact that ranges are disjoint (for all heights~$h$).
\end{proof}

\begin{figure*}[b]
    \centering
    \includegraphics[width=\textwidth]{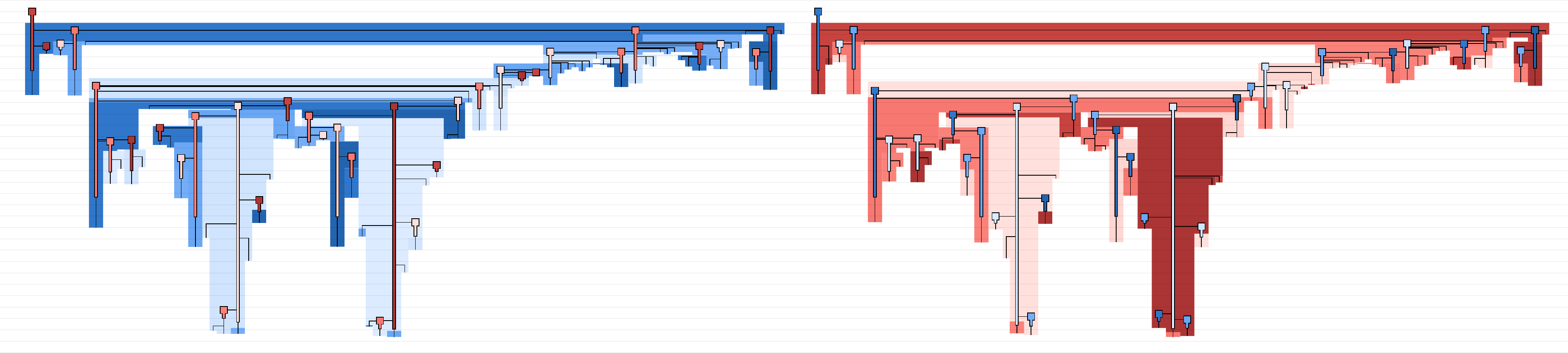}
    \caption{ParkView of a monotone interleaving using six colors of the same hue, with different lightness values, per tree. The merge trees (73 leaves left, 71 leaves right) are derived from Timesteps 177 (left) and 178 (right) of the Ionization Front dataset \cite{pont2022wasserstein}, using a persistence simplification threshold of 0.01. Grid lines at distance of $\delta$.}
    \label{fig:ionization_large1}
\end{figure*}
\begin{figure*}[b]
    \centering
    \includegraphics[width=\textwidth]{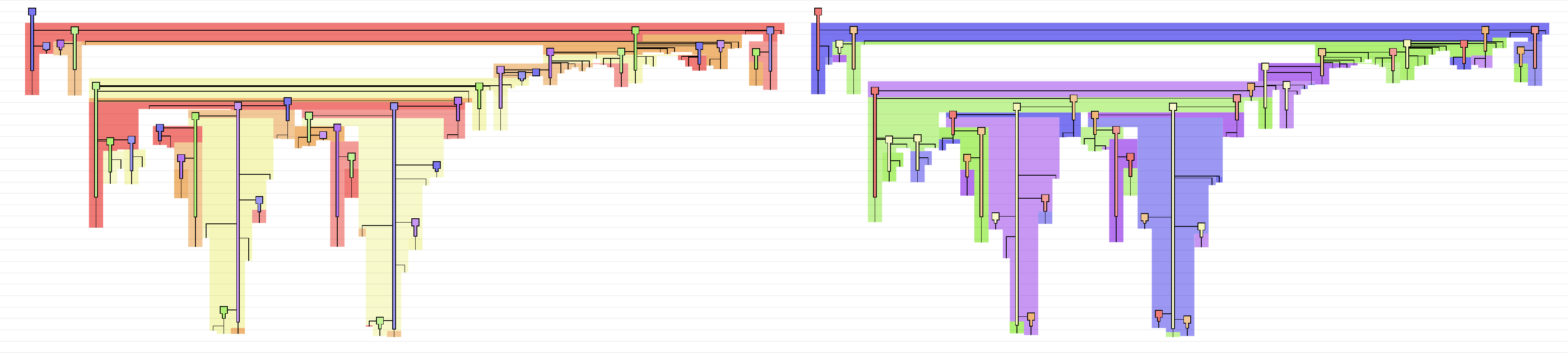}
    \caption{ParkView of a monotone interleaving using six colors, that have distinct hues, per tree. The merge trees (73 leaves left, 71 leaves right) are derived from Timesteps 177 (left) and 178 (right) of the Ionization Front dataset \cite{pont2022wasserstein}, using a persistence simplification threshold of 0.01. Grid lines at distance of $\delta$.}
    \label{fig:ionization_large2}
\end{figure*}

\noindent
Let $\pi$ be a path in a path decomposition $\Pi$ of $T'$.
We say a path $\rho$ in~$\Pi$ is the \emph{parent} of $\pi$ if the top of~$\pi$ lies on~$\rho$.
For a hedge~$H_\pi$, we say that the \emph{top} of~$H_\pi$ is the line segment that forms the top side of the closure of~$H_\pi$. A \emph{gate} of~$H_\pi$ is a non-leaf point~$x$ of $T$ that lies on the top of~$H_\pi$.

\yetanotherproperty*
\begin{proof}  
    Let~$H_\pi$ be a hedge. Unless $H_\pi$ corresponds to the root, it has at least one gate; consider such a gate~$x$. The point~$x$ is contained in another hedge~$H_\rho$, which is a parent of~$H_\pi$. We show that $H_\pi$ cannot have another parent. For this, consider a bar~$b$ adjacent to the top of~$H_\pi$; we show that $b$ is in~$H_\rho$.
    If $b$ is a tree bar, then it contains a gate~$x'$. As any point of~$T$ in the top of a hedge maps to the top of its corresponding path, we know that $\alpha(x') = \alpha(x)$, so $b$ is in~$H_\rho$.
    If $b$ is a filler or bridge, it is surrounded by two tree bars in the same hedge with at least the same height as~$b$. At least one of these two tree bars is adjacent to the top of~$H_\pi$, for if not, the entire top of~$H_\pi$ would be covered by fillers or bridges, and hence there would be no space for the gate. As any such tree bar is part of~$H_\rho$, $b$ is too.
\end{proof}

\begin{observation}
    \label{lem:longest-tree-bar-has-leaf}
    If a tree bar~$b$ is a longest bar in its hedge~$H_\pi$, then $b$ contains a leaf~$\ell$ of~$T$.
\end{observation}
\begin{proof}
    Let $x$ be a lowest point of~$T$ in $b$; as $b$ is a longest bar, $x$ is also a lowest point in~$H_\pi$. We show that~$x$ is the desired leaf~$\ell$ of~$T$.
    Assume for contradiction that $x$ is not a leaf of~$T$, and let $y = \alpha(x)$.
    We now argue that $y$ is an internal vertex of~$T'$, that is, $y$ has more than one down edge.
    By continuity of~$\alpha$, any strict descendant~$x'$ of~$x$ maps to a strict descendant of~$y$. In particular, as~$x$ is not a leaf, there is such a descendant~$x'$ for which $y'$ lies on a down edge~$e$ of~$y$. This means that $y$ is not a leaf, and hence exactly one down edge of~$y$ is contained in~$\pi$. If~$e$ is the only down edge it is thus contained in~$\pi$, meaning that $\alpha(x') \in \pi$ and thus $x' \in B_\pi \subseteq H_\pi$. However, this is impossible as we chose $x$ to be the lowest point in~$H_\pi$. Therefore $y$ is an internal vertex.
    
    As our path decomposition is heavy, $\pi$ contains a heavy down edge~$e$ of~$y$.
    Since $\alpha(x') = y'$ there is at least one down edge of $y$ with weight at least~$1$, and as~$e$ is heavy it thus also has weight at least~$1$.
    So there is a point $x''$ of~$T$ that maps to~$e$.
    As $e$ is contained in~$\pi$, $x''$ must be a point in~$H_\pi$ with $f(x'') < f(x)$, contradicting our choice of~$x$ as the lowest point of~$H_\pi$.
    So $x$ is the desired leaf~$\ell$ of~$T$.
\end{proof}

\anotherpropertyproof*
\begin{proof}
    Let $b$ be a longest bar in a hedge~$H_\pi$. We show that $b$ is a tree bar. A bridge can never be a longest bar in its hedge, so we need to show that $b$ cannot be a filler.
    Assume $b$ is a filler, then by construction there are two tree bars in the same hedge, one left of~$b$ and one right of~$b$, that have the same height as~$b$. By \autoref{lem:longest-tree-bar-has-leaf}, these tree bars contain leaves $\ell_1$ and~$\ell_2$ of~$T$. This implies that there cannot be a leaf in the column of~$b$, because:
    \begin{itemize}
        \item if such a leaf would be below~$b$, then it needs to be connected to the rest of the tree with a horizontal segment that passes below either $\ell_1$ or~$\ell_2$, which violates \autoref{lem:no-leaf-above-horizontal};
        \item if such a leaf would be above~$b$, then the horizontal segment in~$b$ itself would be below that leaf, which again violates \autoref{lem:no-leaf-above-horizontal}.
    \end{itemize}

    So, $b$ is a tree bar. By \autoref{lem:longest-tree-bar-has-leaf}, $b$ contains a leaf~$\ell$. We now show there cannot be a hedge~$H_\rho$ adjacent to the bottom of~$b$. Assume for contradiction that such a hedge~$H_\rho$ does exist. Then $H_\rho$ has at least one gate~$x$, which has to be part of some hedge~$H$ which is a parent of~$H_\rho$. By \autoref{lem:no-multiple-parents}, $H = H_\pi$.

    Let~$b'$ be the bar of~$H_\pi$ containing~$x$. Then $b'$ is a tree bar and, considering $b$ is a longest bar of~$H_\pi$ and adjacent to the top of~$H_\rho$, $b'$ is also a longest bar of~$H_\pi$. Hence, by \autoref{lem:longest-tree-bar-has-leaf}, $x$ is a leaf of~$T$ which contradicts the fact that $x$ is a gate.
\end{proof}

\section{Color Examples}\label{app:colors}
See \autoref{fig:ionization_large1} and \autoref{fig:ionization_large2}.

\section{Large Merge Trees}\label{app:large_example}
See \autoref{fig:volcanic_huge}.
\begin{sidewaysfigure}
    \vspace{8cm}
    \centering
    \includegraphics[width=\textwidth]{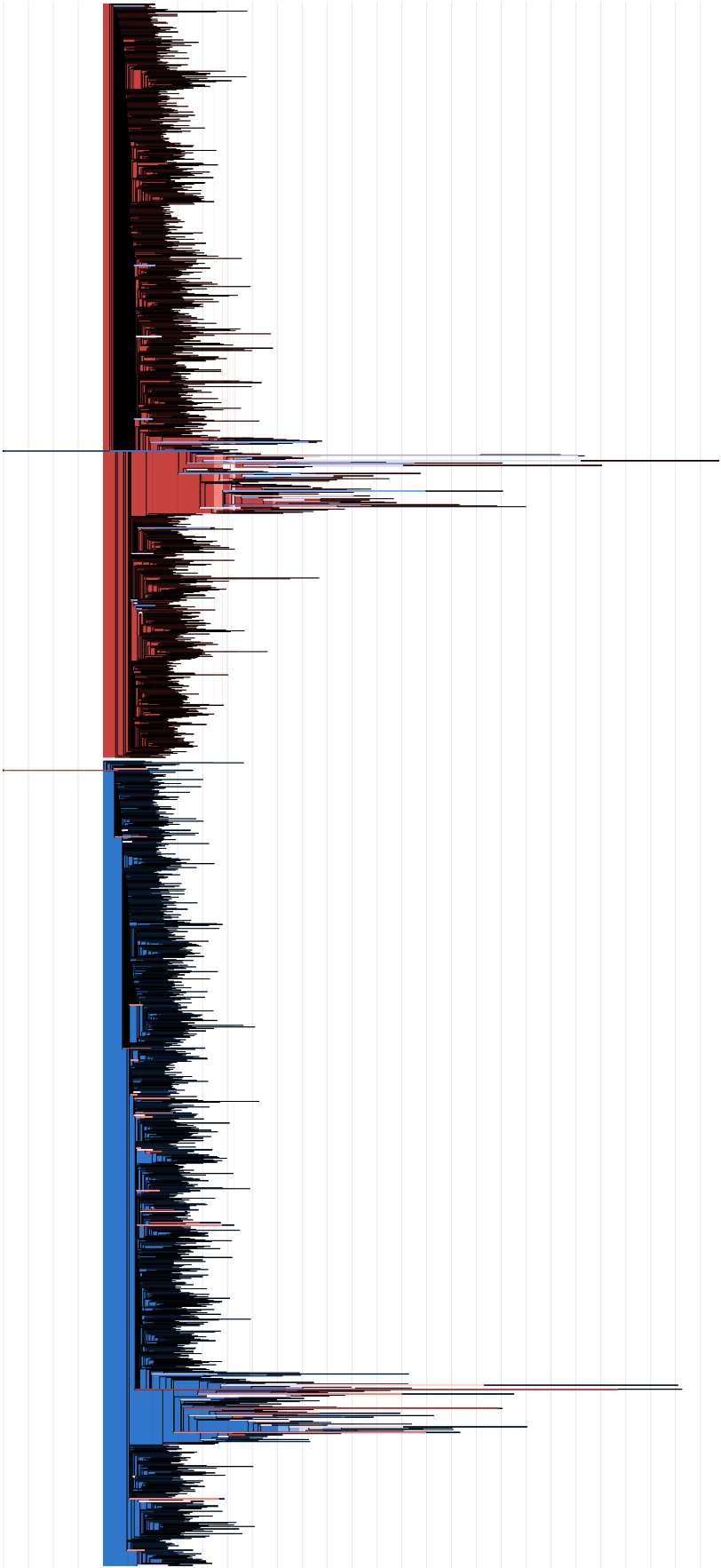}
    \caption{A ParkView visualization of a monotone interleaving between two ordered merge trees derived from time steps 150am (left) and 150pm (right) of the Ionization Front dataset \cite{pont2022wasserstein} using a persistence simplification threshold of 0.01. The left and right tree have 961 and 912 leaves respectively. The grid lines are separated by a distance of $\frac{1}{4}\delta$.}
    \label{fig:volcanic_huge}
\end{sidewaysfigure}

\putbib
\end{bibunit}
\makeatother
\end{document}